\documentclass[letterpaper,aps,onecolumn]{quantumarticle}
\pdfoutput=1

\usepackage{amsmath,amssymb,amsthm,physics,nicefrac}
\usepackage{graphicx}

\usepackage{algorithm}
\usepackage{algorithmicx}
\usepackage[noend]{algpseudocode}
\usepackage{hyperref}

\usepackage{pgfplots}
\pgfplotsset{compat=1.18}
\usepgfplotslibrary{groupplots}

\usepackage{tikz}
\usepackage{subcaption}
\usepackage{booktabs}
\usepackage{siunitx}

\usepackage{float}

\makeatletter
\AddToHook{env/algorithmic/before}{\def\@currentcounter{ALG@line}}
\makeatother

\usepackage{cleveref}
\crefalias{ALG@line}{line}

\title{High-Performance Exact Synthesis of Two-Qubit Quantum Circuits}

\author{Andrew N.~Glaudell}
\affiliation{Photonic Inc., Vancouver, BC, Canada}

\author{Michael Jarret}
\email{mjarretb@gmu.edu}
\affiliation{Department of Mathematical Sciences, George Mason University, Fairfax, VA, USA}
\affiliation{Department of Computer Science, George Mason University, Fairfax, VA, USA}
\affiliation{Quantum Science and Engineering Center, George Mason University, Fairfax, VA, USA}

\author{Swan Klein}
\email{sklein12@gmu.edu}
\affiliation{Department of Mathematical Sciences, George Mason University, Fairfax, VA, USA}

\author{Samuel S.~Mendelson}
\email{smendels@gmu.edu}
\affiliation{Department of Mathematical Sciences, George Mason University, Fairfax, VA, USA}
\affiliation{Department of Computer Science, George Mason University, Fairfax, VA, USA}
\affiliation{Quantum Science and Engineering Center, George Mason University, Fairfax, VA, USA}

\author{T.~C.~Mooney}
\email{tmooney@umd.edu}
\affiliation{Department of Physics, University of Maryland, College Park, MD, USA}
\affiliation{Joint Center for Quantum Information and Computer Science, University of Maryland, College Park, MD, USA}
\affiliation{Joint Quantum Institute, University of Maryland, College Park, MD, USA}

\author{Mingzhen Tian}
\email{mtian1@gmu.edu}
\affiliation{Quantum Science and Engineering Center, George Mason University, Fairfax, VA, USA}
\affiliation{Department of Physics, George Mason University, Fairfax, VA, USA}
\title{High-Performance Exact Synthesis of Two-Qubit Quantum Circuits}

\newcommand{\Z}{\mathbb{Z}}

\newcommand{\SO}{\ensuremath{\operatorname{SO}(6)} }

\newcommand{\swapgate}{\textsc{swap}}

\newcommand{\G}{\mathbb{G}}

\newtheorem{theorem}{Theorem}[section]
\newtheorem{lemma}[theorem]{Lemma}

\newtheorem*{proposition*}{Proposition}

\newtheorem{definition}[theorem]{Definition}

\NewDocumentCommand{\g}{o o}{
    \ensuremath{
      G%
      \IfValueT{#1}{_{#1 \IfValueT{#2}{, #2}}}
  }
}

\NewDocumentCommand{\tcount}{o}{
    \ensuremath{
        \mathcal{T}%
        \IfValueT{#1}{_{#1}}
    }
}
\newcommand{\cliff}{\mathcal{C}_2}

\newcommand{\clifft}{\mathcal{J}_2}

\newcommand{\alglabel}[1]{%
  \begingroup
  \edef\csname @currentlabel\endcsname{%
    \number\csname c@ALG@line\endcsname
  }%
  \label[ALG@line]{#1}%
  \endgroup
}

\usepackage{enumitem}

\newenvironment{loopinvariantproof}
  {\begin{description}[leftmargin=1.5em,style=nextline]}
  {\end{description}}

\newcommand{\Invariant}{\item[\textbf{Invariant.}]}
\newcommand{\Initialization}{\item[\textbf{Initialization.}]}
\newcommand{\Maintenance}{\item[\textbf{Maintenance.}]}
\newcommand{\Termination}{\item[\textbf{Termination.}]}

\begin{document}

\maketitle

\date{Dec 22, 2025}

\begin{abstract}
Exact synthesis provides unconditional optimality and canonical structure, but is
often limited to small, carefully scoped regimes.
We present an exact synthesis framework for two-qubit circuits over the
Clifford+$T$ gate set that optimizes $T$-count exactly.
Our approach exhausts a bounded search space, exploits algebraic
canonicalization to avoid redundancy, and constructs a lookup table of optimal
implementations that turns synthesis into a query.
Algorithmically, we combine meet-in-the-middle ideas with provable pruning rules
and problem-specific arithmetic designed for modern hardware.
The result is an exact, reusable synthesis engine with substantially improved
practical performance.
\end{abstract}

\section{Introduction}

Exact quantum circuit synthesis serves a different purpose from heuristic
optimization.
Rather than producing a good circuit quickly, it aims to completely characterize
a well-defined region of circuit space: which unitaries are reachable and what optimal circuits look like under a
chosen cost metric.
When such a characterization is possible, it yields canonical representatives, reusable building blocks, and a comparison
for evaluating heuristic quality.

Recent work shows that this approach remains valuable even when the general
problem is intractable.
Bravyi, Latone, and Maslov demonstrate this clearly for Clifford circuits:
they exploit symmetry and canonicalization to exhaustively map circuit space up to
six qubits, store optimal representatives in a database, and support fast
extraction by query \cite{bravyi20226}.
Optimal synthesis may not scale well, but it can be pre-computed circuits can be used as a resource for other methods.

In fault-tolerant compilation, the dominant resource cost is typically associated with non-Clifford operations. In the Clifford+$T$ gate set, this motivates optimizing \emph{T-count}. Although exact $T$-count optimization is computationally intractable in general \cite{van2023optimising}, compilers repeatedly encounter small subcircuits and canonical building blocks whose optimal realizations recur across many instances. This creates an opportunity for \emph{offline exact synthesis}: exhaustively enumerate a structured regime once, store optimal representatives up to natural equivalences, and reduce subsequent synthesis to fast database queries.

We exhaust a bounded search space, canonicalize targets to avoid redundant work,
and construct a lookup table (LUT) of optimal circuits. The output is a database of
optimal two-qubit Clifford+$T$ circuits that can be reused across
compilation tasks. Any query to our database returns a circuit with provably minimal $T$-count for its target unitary.
One can use the LUT directly for optimization, embed it as a subroutine inside
larger compilers, or continually generate results on a cluster to map out large regions of the
space ahead of time.

Although we instantiate everything in Clifford+$T$, we view the current approach as more
general than that choice.
The algorithm has two components: (1) an outer shell that exhaustively generates shortest circuits,
puts representatives into canonical form, and performs meet-in-the-middle style joining \cite{Amy2013}, and (2) a high performance backbone that
implements gate-set-specific algebra. 

Because of its ubiquitous nature, 
Clifford+$T$ optimized over $T$-count is a good first gate set.
Nonetheless, the same ``pay once, query forever'' strategy applies to other discrete
gate sets and other resource counts, as long as one can represent elements
exactly, define a cost measure, and decide exact equivalence in the intended
model. In fact, we realized that some of the heavier mathematical machinery below may have been an artifact of analysis and a more abstract, general approach might be more efficient.

Our goal is not exact optimization at large scale, nor replacing heuristic compilers, but rather we intend to provide reusable, exact infrastructure for a structured two-qubit regime that can be embedded as a backend inside larger compilation pipelines. Throughout this work, we emphasize correctness, reuse, and performance within this scope rather than asymptotic scalability.

Our performance depends heavily on our algebraic backbone: by replacing general-purpose linear algebra with gate-set-aware arithmetic and linear maps, we are able to gain many orders of magnitude in efficiency. That is, we rebuild every layer of the computation from the bottom up using customized data structures. This makes the problem far more tractable than previously seemed conceivable.

Recent results make clear why scoped, structure-driven approaches are the
right level of ambition.
Van de Wetering and Amy show that minimizing $T$-count or $T$-depth over
Clifford+$T$ is NP-hard in the worst case, even in the decision form ``does there
exist an equivalent circuit with at most $k$ T gates?'' \cite{van2023optimising}.
They also place the problem in $\mathrm{NP}^{\mathrm{NQP}}$, because verification naturally requires deciding \emph{exact} circuit equivalence
rather than approximate closeness \cite{van2023optimising}.
Kjelstrøm, Pavlogiannis, and van de Pol strengthen the lower bound further:
they prove co-NQP-hardness for a broad family of exact circuit optimization
tasks, and their corollary applies directly to Clifford+$T$, where minimizing
non-Clifford gates coincides with minimizing $T$-count \cite{Kjelstrom}.
Taken together, these results rule out fast, always-exact, general-purpose
optimization at scale (unless widely believed complexity-theoretic separations
collapse).

Although these hardness results limit our expectations, they do not diminish the value of exact
synthesis in regimes where completeness and reuse matter more than asymptotic
scalability.
Compilers repeatedly encounter small subcircuits, gadget fragments, and canonical
blocks whose optimal implementations recur across many instances.
In this setting, offline enumeration can pay for itself: we do the expensive work
once, and we reuse the result indefinitely. Since the same subsequence might be compiled multiple times, a reference-able database or efficient data structure can lead to improved circuitry and compile times.

The technical contribution of this paper is an exact synthesis engine that makes this amortized approach practical for two-qubit Clifford+$T$ circuits. Algorithmically, we build on meet-in-the-middle ideas \cite{Amy2013}, but redesign enumeration around canonicalization, equivalence relations, and reuse, rather than treating these as secondary concerns. Representationally, we replace general-purpose linear algebra with gate-set-aware arithmetic in an $\SO$ model, so that each $T$-count can be implemented using specialized constant-time kernels. This shifts the practical bottleneck from avoidable constant factors to the intrinsic growth of the quotient search space, enabling exact $T$-count synthesis at depths that are impractical for baseline exact-synthesis implementations on comparable hardware.

\paragraph{Contributions.}
This paper:
(i) presents an exact two-qubit Clifford+$T$ synthesis procedure with provably
minimal $T$-count within the enumerated domain;
(ii) gives canonicalization and pruning rules with full correctness proofs;
(iii) describes a high-performance implementation based on problem-specific
arithmetic and HPC-aware design; and
(iv) reports empirical scaling results, including direct benchmarks against
meet-in-the-middle baselines.

The remainder of the paper is organized as follows.
\Cref{sec:Preliminaries} fixes notation, the synthesis model, and the definition
of $T$-count.
\Cref{sec:theory} describes the structure of the search space and the
reductions that keep enumeration finite.
\Cref{sec:algorithms} presents the algorithm, pseudocode, and proofs of
correctness.
\Cref{sec:Results} reports performance and scaling results.
We conclude with limitations and directions for extending this approach.

\section{Preliminaries}\label{sec:Preliminaries}

\subsection{Clifford+T synthesis}\label{sub:background}

Quantum algorithms are specified at a high level, either as abstract
unitaries or as circuits over idealized gate sets.
In practice, both hardware constraints and fault-tolerant error correction restrict
the available operations.
\emph{Quantum circuit synthesis} is the task of translating an abstract description
of a unitary into a circuit over a specified gate set.
When the target unitary can be represented exactly, this task is referred to as
\emph{exact synthesis}; otherwise, one considers \emph{approximate synthesis}, where
the goal is to approximate the target to within a prescribed tolerance.

The Clifford+$T$ gate set plays a central role in fault-tolerant quantum computing.
It is universal for quantum computation, while also admitting efficient error
correction schemes for the Clifford subgroup.
A common generating set is
\[
\{H, S, CX, T\},
\]
where $H$ is the Hadamard gate, $S$ is the phase gate, $CX$ is the controlled-NOT
gate, and $T=\sqrt{S}$ is the $\pi/8$ rotation. Explicitly,
\[
H = \frac{1}{\sqrt{2}}
\begin{bmatrix}
1 & 1 \\
1 & -1
\end{bmatrix}, \quad
S =
\begin{bmatrix}
1 & 0 \\
0 & i
\end{bmatrix}, \quad
CX =
\begin{bmatrix}
1 & 0 & 0 & 0 \\
0 & 1 & 0 & 0 \\
0 & 0 & 0 & 1 \\
0 & 0 & 1 & 0
\end{bmatrix}, \quad
T =
\begin{bmatrix}
1 & 0 \\
0 & e^{i\pi/4}
\end{bmatrix}.
\]
The gates $H$, $S$, and $CX$ generate the Clifford group, which is not universal.
The $T$ gate is the only non-Clifford gate in this set, and in most fault-tolerant
architectures it dominates the resource cost of a computation.
For this reason, a large body of work focuses on minimizing the \emph{$T$-count} of
Clifford+$T$ circuits.

Much of the existing literature on exact Clifford+$T$ synthesis is based on a
number-theoretic characterization of these circuits
\cite{giles2013exact,gosset2014algorithm,kliuchnikov2012fast,
kliuchnikov2013synthesis,selinger2012efficient,kliuchnikov2024stabilizer,
Kliuchnikov2024multi,li2023fast}.
In this framework, Clifford+$T$ unitaries are represented using matrices over
rings of cyclotomic integers, and synthesis proceeds by reducing a numerical
invariant—often the \emph{least denominator exponent}—that correlates with $T$-count.
For single-qubit circuits, this approach yields optimal exact synthesis algorithms \cite{giles2013exact},
and for two qubits it leads to asymptotically optimal methods \cite{li2023fast}.

While the least denominator exponent is a powerful organizing principle, it does
not provide a monotone proxy for $T$-count beyond the single-qubit setting.
Even for two-qubit circuits, the least denominator exponent can decrease while
$T$-count increases.
This complicates its direct use as a control parameter for exhaustive enumeration
or database construction.

In this work, we take a different approach.
Rather than using number-theoretic invariants to guide a single synthesis trajectory,
we use the underlying algebraic structure to \emph{enumerate} Clifford+$T$ circuits
by $T$-count and store optimal results for later reuse.
Our goal is not to synthesize one circuit at a time, but to construct a database
of optimal two-qubit Clifford+$T$ circuits up to a prescribed $T$-count.

The core idea is to enumerate circuits recursively, organize them into equivalence
classes under exact Clifford+$T$ equivalence, and store a single optimal
representative for each class.
This dramatically reduces the size of the database compared to storing all circuits
explicitly.
The tradeoff is that identifying canonical representatives is computationally more
expensive than raw enumeration, and it becomes one of the dominant costs in database
generation.

Two ingredients make this approach practical.
First, we adopt a nontraditional generating set for Clifford+$T$ circuits that
facilitates efficient enumeration by $T$-count \cite{li2023fast}.
Second, we define an equivalence relation tailored to two-qubit Clifford+$T$ that
allows us to prune redundant search paths without sacrificing optimality.
The resulting database can be used in several
ways: as a direct lookup table, as the basis for meet-in-the-middle synthesis, or
as a backend for sliding-window and peephole optimization. Without aggressive reduction, storing all two-qubit Clifford+$T$ circuits quickly
becomes infeasible.
With canonicalization, exhaustive enumeration up to nontrivial $T$-count bounds
becomes practical on modern hardware.

\section{Theory}\label{sec:theory}

This section introduces two tools that drive our enumeration algorithm.
First, we represent two-qubit Clifford+$T$ circuits inside \SO, where
multiplication and comparison are cheap.
Second, we quotient by left and right Clifford multiplication, so that we store
only one representative per equivalence class. We ultimately show that one representative per equivalence class is sufficient to generate all $T$-optimal Clifford+$T$ circuits of a given $T$-count.

\subsection{Notation}
Let $R$ be a ring.
We write $\mathrm{U}_n(R)$ for the group of $n\times n$ unitary matrices with
entries in $R$, $\mathrm{SU}_n(R)$ for the determinant-one subgroup, and
$\mathrm{SO}_n(R)$ for the real special orthogonal group.
When $R=\mathbb{C}$ or $\mathbb{R}$ we abbreviate $\mathrm{U}_n(\mathbb{C})=\mathrm{U}(n)$,
$\mathrm{SU}_n(\mathbb{C})=\mathrm{SU}(n)$, and $\mathrm{SO}_n(\mathbb{R})=\mathrm{SO}(n)$.

Let $\mathcal{J}_n\subset \mathrm{U}(2^n)$ be the subgroup generated by $n$-qubit
Clifford+$T$ circuits, and let $\mathcal{C}_n\subset \mathcal{J}_n$ be the
$n$-qubit Clifford group.
For $k\ge 0$, let $\tcount[k]\subset \mathcal{J}_2$ denote the set of two-qubit
unitaries with \emph{optimal} T-count equal to $k$.

We also use $[n]=\{1,\dots,n\}$ and let $e_1,\dots,e_n$ be the standard basis of
$\mathbb{R}^n$.
A \emph{homomorphism} is a map-preserving multiplication, and an $m$-to-one
\emph{cover} is a surjective homomorphism for which each fiber has cardinality $m$.

\subsection{Equivalence classes}
\label{sub:equiv_classes}
Our goal in this section is to formalize equivalences of quantum circuits. Intuitively, if two circuits $U_0,U_1$ act on the same (normalized) quantum state $\ket{\psi}$ such that $\abs{\bra{\psi}U_0^\dagger U_1 \ket{\psi}}=1$, we claim that these states are physically indistinguishable and so any natural definition of equivalence should have that $U_1 \in [U_0]$. Although any equivalence relation should retain this property, there are other forms of equivalence one might introduce and exploit. For instance, although $\swapgate \cdot U \cdot \swapgate \neq U$ for general $U$, one can readily see that these two circuits are ``equivalent'' in some logical sense. In fact, since $\swapgate \in \cliff$, we trivially have that $\swapgate \cdot U \cdot \swapgate U \sim_\mathcal{C} U $.

Quantum circuits are defined only up to global phase: multiplying a circuit by
a scalar of unit modulus does not change the physical operation it implements.
In addition, for compiling it is natural to identify circuits that
differ by left or right multiplication by Clifford operations, since such
differences correspond to basis changes that are inexpensive to implement and
do not affect non-Clifford cost.

Accordingly, throughout this work we treat two Clifford+$T$ circuits as
\emph{equivalent} if they represent the same unitary up to global phase and
Clifford conjugation.
Our goal is not to distinguish between equivalent circuits, but to store a
single optimal representative for each equivalence class.
The enumeration and pruning procedures developed later are designed to respect
this equivalence from the outset.

To make this equivalence computationally tractable, we work with a representation
of two-qubit Clifford+$T$ circuits inside \SO.
This representation preserves exact equivalence while enabling fast matrix
operations and explicit bookkeeping of Clifford actions.

\subsection{The SO(6) representation}

\label{sub:so6}
We begin with a standard representation of $\mathrm{SU}(4)$ in \SO.

\begin{theorem}\label{thm:except}
There exists a two-to-one cover $\Phi:\mathrm{SU}(4)\to \SO$.
For $U\in\mathrm{SU}(4)$ we write $\overline{U}=\Phi(U)$.
In particular, $\overline{U}=\overline{-U}$ and
$\overline{iU}=\overline{-iU}=-\overline{U}$.
\end{theorem}

Theorem~\ref{thm:except} is a known consequence of the exceptional
isomorphism $\mathrm{Spin}(6)\cong \mathrm{SU}(4)$.
We use it as a computational representation: it replaces $4 \times 4$ complex matrices with $6\times 6$
real orthogonal matrices. This allows us to compute products of reals rather than introducing methods to handle complex numbers.

Since two Clifford+$T$ circuits are equivalent if they differ only by global phase, we extend
$\overline{(\cdot)}$ from $\mathrm{SU}(4)$ to $\mathrm{U}(4)$ by choosing an
arbitrary determinant-one representative.
Concretely, for $U\in\mathrm{U}(4)$ we pick a fourth root of $\det(U)$ and set
\[
U' = U/\det(U)^{1/4}\in \mathrm{SU}(4),
\qquad
\overline{U} := \overline{U'}.
\]
This choice is not canonical: different fourth roots differ by a factor in
$\{1,i,-1,-i\}$, and by Theorem~\ref{thm:except} these choices map to at most two
images in \SO (differing by an overall sign).
Nothing in our algorithm depends on which representative is chosen, because our
equivalence relation ignores global phase and all computations take place inside
$\overline{\mathcal{J}_2}\subset \SO$.
For a set $S\subset \mathrm{U}(4)$ we write $\overline{S}=\{\overline{U}:U\in S\}$.

The usefulness of \cref{thm:except} for compilation was first seen in \cite{glaudell2021optimal} and later in \cite{li2023fast}. Recall, $\cliff$ is generated by the Hadamard gate $H$, the phase gate $S$, and the controlled-$Z$ gate $CZ$ and $\clifft$ is generated by adding the gate $T=\sqrt{S}$. Under the map defined above, these matrices are given by:
\[\overline{H\otimes I} = \begin{bmatrix}
    0 & 0 & 1 & 0 & 0 & 0\\
    0 & -1 & 0 & 0 & 0 & 0\\
    1 & 0 & 0 & 0 & 0 & 0\\
    0 & 0 & 0 & 1 & 0 & 0\\
    0 & 0 & 0 & 0 & 1 & 0\\
    0 & 0 & 0 & 0 & 0 & 1\\
\end{bmatrix},\overline{S\otimes I} = \begin{bmatrix}
    0 & -1 & 0 & 0 & 0 & 0\\
    1 & 0 & 0 & 0 & 0 & 0\\
    0 & 0 & 1 & 0 & 0 & 0\\
    0 & 0 & 0 & 1 & 0 & 0\\
    0 & 0 & 0 & 0 & 1 & 0\\
    0 & 0 & 0 & 0 & 0 & 1\\
\end{bmatrix},\overline{T\otimes I} = \begin{bmatrix}
    \frac{1}{\sqrt{2}} & -\frac{1}{\sqrt{2}} & 0 & 0 & 0 & 0\\
    \frac{1}{\sqrt{2}} & \frac{1}{\sqrt{2}} & 0 & 0 & 0 & 0\\
    0 & 0 & 1 & 0 & 0 & 0\\
    0 & 0 & 0 & 1 & 0 & 0\\
    0 & 0 & 0 & 0 & 1 & 0\\
    0 & 0 & 0 & 0 & 0 & 1\\
\end{bmatrix}\]
\[\overline{I\otimes H} = \begin{bmatrix}
    1 & 0 & 0 & 0 & 0 & 0\\
    0 & 1 & 0 & 0 & 0 & 0\\
    0 & 0 & 1 & 0 & 0 & 0\\
    0 & 0 & 0 & 0 & 0 & 1\\
    0 & 0 & 0 & 0 & -1 & 0\\
    0 & 0 & 0 & 1 & 0 & 0\\
\end{bmatrix},\overline{I\otimes S} = \begin{bmatrix}
    1 & 0 & 0 & 0 & 0 & 0\\
    0 & 1 & 0 & 0 & 0 & 0\\
    0 & 0 & 1 & 0 & 0 & 0\\
    0 & 0 & 0 & 0 & -1 & 0\\
    0 & 0 & 0 & 1 & 0 & 0\\
    0 & 0 & 0 & 0 & 0 & 1\\
\end{bmatrix},\overline{I\otimes T} = \begin{bmatrix}
    1 & 0 & 0 & 0 & 0 & 0\\
    0 & 1 & 0 & 0 & 0 & 0\\
    0 & 0 & 1 & 0 & 0 & 0\\
    0 & 0 & 0 & \frac{1}{\sqrt{2}} & -\frac{1}{\sqrt{2}} & 0\\
    0 & 0 & 0 & \frac{1}{\sqrt{2}} & \frac{1}{\sqrt{2}} & 0\\
    0 & 0 & 0 & 0 & 0 & 1\\
\end{bmatrix},\]
\[\overline{CZ} = \begin{bmatrix}
    0 & -1 & 0 & 0 & 0 & 0\\
    1 & 0 & 0 & 0 & 0 & 0\\
    0 & 0 & 0 & 0 & 0 & -1\\
    0 & 0 & 0 & 0 & -1 & 0\\
    0 & 0 & 0 & 1 & 0 & 0\\
    0 & 0 & 1 & 0 & 0 & 0\\
\end{bmatrix}.\]

Through direct computation we see $\overline{\cliff}=\SO(\{-1,0,1\})$, the set of \emph{signed permutation matrices} with determinant 1, and $\overline{\clifft}\subset \text{SO}_6\left(\Z[1/\sqrt{2}]\right)$ where 
\[
\Z\left[\frac{1}{\sqrt{2}}\right] = \left\{\left.\frac{a+b\sqrt{2}}{\sqrt{2}^c}~\right|~a,b,c\in\Z,c\geq 0\right\}.
\] In fact, $\overline{\clifft} = \text{SO}_6\left(\Z[1/\sqrt{2}]\right)$ \cite{li2023fast}. Because the kernel of this map is contained in the center of $\clifft$ (those operations that commute with all Clifford-$T$ circuits), and the center of $\clifft$ is contained in $\cliff$, we can always introduce the correct global phase using a Clifford circuit when mapping back to $\clifft$ from $\overline{\clifft}$.
With these observations in mind, we define a new set of generators $\G$.

\begin{definition}
    Let $P(i,j)\in\text{SO}(6)$ be the signed permutation matrix such that $P(i,j)e_i = -e_j$, and $P(i,j)e_k = e_k$ if $k\neq i,j$. Let $\g[i,j] = P(i,j)^\dagger\left(\overline{T\otimes I}\right)P(i,j)$ and define $M = \{\g[i,j]~|~ 1\leq i, j\leq 6,i\neq j\}$.
\end{definition}

Each $\g[i,j]$ can be written as the image of $\overline{C^\dagger (T\otimes I) C}$ for some Clifford operator $C$. Each $\g[i,j]$ is a rearrangement of $\overline{T\otimes I}$ with the first and $i$-th rows and columns swapped and the second and $j$-th rows and columns swapped. As an example, explicitly \[
\g[2,4] = \begin{bmatrix}
    1 & 0 & 0 & 0 & 0 & 0\\
    0 & \frac{1}{\sqrt{2}} & 0 & -\frac{1}{\sqrt{2}} & 0 & 0\\
    0 & 0 & 1 & 0 & 0 & 0\\
    0 & \frac{1}{\sqrt{2}} & 0 & \frac{1}{\sqrt{2}} & 0 & 0\\
    0 & 0 & 0 & 0 & 1 & 0\\
    0 & 0 & 0 & 0 & 0 & 1\\
\end{bmatrix}.
\]

We show in the following lemma that $\G$ is in fact a generating set for $\overline{\clifft}$.

\begin{lemma}\label{lemma:generators}
    The set $\G$ generates $\overline{\clifft}$.
\end{lemma}
\begin{proof}
    By definition $M\subset \overline{\clifft}$. It then suffices to show that $\overline{T\otimes I}$, $\overline{S\otimes I}$, $\overline{H\otimes I}$, $\overline{I\otimes T}$, $\overline{I\otimes S}$, $\overline{I\otimes H}$, and $\overline{CZ}$ can be generated by $\G$.

    First, $\overline{T\otimes I} = \g[1,2]$ and $\overline{I\otimes T} = \g[4,5]$ and consequently $\overline{S\otimes I} = \g[1,2]^2$ and $\overline{I\otimes S} = \g[4,5]^2$. Next, $\overline{H\otimes I} = \g[1,3]^2\g[2,3]^4$ and $\overline{I\otimes H} = \g[4,6]^2M_[5,6]^4$. Finally, $\overline{CZ} = \g[1,2]\g[3,6]\g[4,5]$.
\end{proof}

We require one more technical lemma before our final result.

\begin{lemma}\label{lem:tech}
    Let $P\in\text{SO}(6)$ be a signed permutation matrix and let $\sigma_P:[6]\rightarrow \Z$ with $\sigma(i) = j$ if $Pe_i = e_j$ and $\sigma(i) = -j$ if $Pe_i = -e_j$.
    \begin{itemize}
        \item[1.] Let $C\in\cliff$, Then \[
        \overline{C}\g[i,j] = \begin{cases}
        \g[|\sigma_{\overline{C}}(i)|,|\sigma_{\overline{C}}(j)|] \overline{C} & \text{if }\sigma_{\overline{C}}(i)\sigma_{\overline{C}}(j) > 0\\
        \g[|\sigma_{\overline{C}}(j)|,|\sigma_{\overline{C}}(i)|]\overline{C} & \text{if }\sigma_{\overline{C}}(i)\sigma_{\overline{C}}(j) < 0
        \end{cases}
        \]
        \item[2.] If $U\in\tcount[k]$, then $\overline{U}$ can be written using a signed permutation matrix and $k$ operators from $\G$.
    \end{itemize}
\end{lemma}

\begin{proof}
    1. can be proven by direct computation.

    Let $U\in\tcount[k]$. Then $U = C_1(T\otimes I)C_2(T\otimes I)C_3\dots C_k (T\otimes I) C_{k+1}$ for some $C_i\in\cliff$ and \[
    \overline{U} = \overline{{C_1}{(T\otimes I)}{C_2}\dots {C_k} {(T\otimes I)} {C_{k+1}}} = P_1\g[1,2]P_2\dots P_k \g[1,2] P_{k+1}
    \] where each $P_i$ is a signed permutation matrix. Using the commutation rules of 1. we can shift each $P_i$ to the right side of the equation and finally \[
        \overline{U} = \g[i_1,j_1] \g[i_2,j_2] \dots \g[i_k,j_k] P.
    \]
\end{proof}

To end this section we present an equivalence relation on $\clifft$ and show that a single representative of each class from $\tcount[k]$ can be used to compute a representative for each class of $\tcount[k+1]$.

\begin{theorem}\label{thm:equiv}
    Let $U,V\in\clifft$. Define an equivalence relation $\sim$ on $\clifft$ where $U\sim V$ if and only if $C_1 U = V C_2$ for some $C_1,C_2\in \cliff$. The equivalence relation $\sim$ induces an equivalence relation on $\overline{\clifft}$. Let $\widetilde{\tcount[k]}\subset \tcount[k]$ be a set containing at least one representative for each equivalence class of $\tcount[k]$. If $S\in\tcount[k+1]$, then $\overline{S}$ is equivalent to at least one element of $\{\overline{T}\g[i,j]\mid1\leq i,j\leq 6,T\in\widetilde{\tcount[k]}\}$.
\end{theorem}

\begin{proof}
    Let $S\in\tcount[k+1]$. Then $\overline{S} \sim \g[i_1,j_1]\g[i_2,j_2]\dots\g[i_{k+1},j_{k+1}]$ for $\g[i_m,j_m]\in\G$ by \cref{lem:tech}. By hypothesis, $\g[i_1,j_1] \g[i_2,j_2] \dots \g[i_k,j_k] \sim \overline{T}$ for some $ T\in\widetilde{\tcount[k]}$. Let $\overline{T} = \overline{D_1} \g[p_1,q_1] \g[p_2,q_2] \dots \g[p_k,q_k] \overline{D2}$ for some $D_1,D_2\in\cliff$. So $\overline{S} \sim \overline{D_1}\g[p_1,q_1] \g[p_2,q_2] \dots \g[p_k,q_k]\overline{D_2} \g[i_{k+1},j_{k+1}]$. By \cref{lem:tech}, \[
    \overline{D_1}\g[p_1,q_1] \g[p_2,q_2] \dots \g[p_k,q_k] \overline{D_2} \g[i_{k+1},j_{k+1}] = \overline{D_1}\g[p_1,q_1] \g[p_2,q_2] \dots \g[p_k,q_k] \g[p_{k+1},q_{k+1}]\overline{D_2}
    \] for some $p_{k+1},q_{k+1}$. 
    Therefore, $\overline{S}\sim \overline{T}\g[p_{k+1},q_{k+1}]$.
\end{proof}

\Cref{thm:equiv} allows us to compute a representative for each equivalence class of Clifford-$T$ circuits of $T$-count $k+1$ from a set of single representatives for each equivalence class of $T$-count $k$, significantly reducing the space required to store optimal circuits.

\subsection{An involutive generating set for representatives}
\label{sub:involutive_generators}

Recall that we work inside the matrix group
$\mathcal{J}_2 \subset \SO$ and that our set of
$T$-step generators are
\[
G=\{G_{i,j} : 1\le i\ne j\le 6\},
\qquad
G_{i,j}=P(i,j)^\dagger \,\overline{(T\otimes I)}\,P(i,j),
\]
where $P(i,j)$ is a signed permutation.
By \Cref{lemma:generators}, this set generates
$\mathcal{J}_2$.

Our goal is only to find a representative subgroup and, as such, we do not need \emph{all} generators in $G$
as explicit step operators. Additionally, our equivalence classes extend beyond $SO(6)$ and into $O(6)$. The proofs of the previous subsection do not depend specifically on the restriction to \SO and, hence, we proceed by relaxing this constraint in order to consider an involutive set of representative generators.

The recursion in \Cref{thm:equiv}
only requires that, from one representative per equivalence class at $T$-count $k$,
we can generate a representative per class at $T$-count $k+1$ by multiplying by a
single $T$-step. This subsection shows that we can
choose those $T$-steps from a simpler set in which every generator is an
\emph{involution}, ultimately generating a smaller group. This matters both theoretically (it clarifies the bipartite
structure exploited by the algorithms) and practically (it simplifies compiling).

For each ordered pair $i\ne j$, let $C_{i,j}$ be the permutation matrix that swaps rows $i$ and $j$. Then, we know by \Cref{thm:equiv} that $C_{i,j}G_{i,j} \sim G_{i,j}$ and, moreover, that the set $\mathcal{G} = \{C_{i,j}G_{i,j} \mid 1\leq i < j \leq 6\}$ is sufficient to generate a representative for every equivalence class in $\mathbb{G}$. Although we claim this is obvious, we formalize it in the following theorem.
\begin{theorem}\label{thm:involution}
    Every element $g \in \clifft$ is equivalent to a product of involutions.
\end{theorem}

\begin{proof}
    It suffices to show that $\G \sim \G'$ where $\G'$ contains only involutions. To see this, note that for each $G_{i,j} \in \G$, $G_{i,j}\sim C_{i,j} G_{i,j}$ and $C_{i,j}G_{i,j}$ can be easily seen to be an involution. Hence, the set $\G' = \{C_{i,j}G_{i,j} \mid G_{i,j} \in \G\}$ contains only involutions. Moreover, by \Cref{thm:equiv}, we know that for all $\tcount[k]$, $\tcount[k]C_{i,j} \sim \tcount[k]$. Hence, it follows immediately from \Cref{thm:equiv} that $\tcount[k]\G' \sim \tcount[k]\G$, completing the proof.
\end{proof}

Although \Cref{thm:involution} in principle dramatically reduces the search space, its application is currently restricted to accelerating our codebase. Involutions are, after all, easier to work with. Future work might exploit the involutive property to discover better canonical forms or algorithms; for present purposes we only use this property practically.

\section{Optimal Exact Synthesis Algorithm}\label{sec:algorithms}

In this section we give pseudocode for the breadth-first enumeration that drives
LUT generation. Using the \SO representation from
Section~\ref{sub:so6}, the problem reduces to exploring the Cayley graph of
$\overline{\mathcal{J}_2}$ generated by $G$, modulo our canonicalization
procedure. The loop structure below is mildly inelegant, but it parallelizes
cleanly.

\algnewcommand{\LineComment}[1]{%
  \Statex \noindent \textcolor{gray}{\texttt{//}~#1}
}

\begin{algorithm}
\caption{Generate the lookup table up to T-distance $k$}
\label{alg:LUT}
\begin{algorithmic}[1]
\Function{GenerateLUT}{$r \in \SO$, $k\in\mathbb{Z}_{\ge 0}$}
    \State $\textproc{LUT} \gets \Call{InitLUT}{\{r\}}$
    \For{$i=0$ to $k-1$}
        \State \Call{ExtendOneStep}{$\textproc{LUT}$}
    \EndFor
    \State \Return $\textproc{LUT}$
\EndFunction

\vspace{0.8em}

\LineComment{\textit{\textproc{LUT.back} and \textproc{LUT.previous} refer to the final and penultimate elements of \textproc{LUT} respectively}}
\LineComment{\textit{\textproc{next} is a set of representatives}} \LineComment{\textit{\textproc{next.insert} inserts only when no equivalent member exists}}
\Function{ExtendOneStep}{$\textproc{LUT}$}
    \State $\textproc{next} \gets \emptyset$ 
    \ForAll{$U \in \textproc{LUT.back}$}
        \ForAll{$g \in G$} 
            \State $n \gets g\cdot U$
            \If{$n \notin \textproc{LUT.previous}$}\label{line:bipartite-start} \Call{next.insert}{$n$} \label{line:bipartite-end}  
            \EndIf 
        \EndFor
    \EndFor
    \State $\Call{LUT.push\_back}{next}$
\EndFunction

\vspace{0.8em}

\Function{InitLUT}{$r\in \SO$}
    \State Declare a vector $\textproc{LUT}$
    \State \Call{LUT.push\_back}{$\{r\}$} 
    \State \Return $\textproc{LUT}$
\EndFunction

\end{algorithmic}
\end{algorithm}

\begin{algorithm}
\caption{Meet-in-the-middle search \label{alg:mitm}}
\begin{algorithmic}[1]

\Function{MITM}{$r_L \in \SO, r_R \in \SO$}
    \If{$r_L = r_R$} \Return $e$ \EndIf
    \State $\textproc{left} \gets$ \Call{InitLUT}{$r_L$}
    \State $\textproc{right} \gets$ \Call{InitLUT}{$r_R$}
    \State $\textproc{meet} \gets \bot$
    \While{$\textproc{meet} = \bot$}    
        \If{$\abs{\textproc{left.back}} < \abs{\textproc{right.back}}$}
            \State $\textproc{meet} \gets \Call{ExtendOneStep\_MITM}{\textproc{left},\textproc{right}}$
        \Else
            \State $\textproc{meet} \gets \Call{ExtendOneStep\_MITM}{\textproc{right},\textproc{left}}$
        \EndIf
    \EndWhile
    \State $m_L \gets \Call{left.find}{\textproc{meet}}$ \Comment{$l \sim \textproc{meet}$}
    \State $m_R \gets \Call{right.find}{\textproc{meet}}$ \Comment{$r \sim \textproc{meet}$}
    \State Find $c_0,c_1 \in \cliff$ such that $c_1 m_L r_L c_0 = m_R r_R$. \Comment{Do however. Brute force is fine here.}
    \State \Return the word $m_R^{-1} c_1 m_L r_L$
\EndFunction
\vspace{0.8em}

\Function{ExtendOneStep\_MITM}{\textproc{to\_extend},\textproc{to\_check}}
    \State $\textproc{next} \gets \emptyset$
    \ForAll{$U \in \textproc{to\_extend.back}$}
        \ForAll{$g \in G$}
            \State $n \gets g\cdot U$
            \If{$n \notin \textproc{to\_extend.back}$}
                \State $\Call{next.insert}{n}$
                \If{$n \in \textproc{to\_check.back}$} \Return $n$ \EndIf
            \EndIf
        \EndFor
    \EndFor
    \State \Return $\bot$
\EndFunction

\end{algorithmic}
\end{algorithm}

The following theorems will justify \Cref{line:bipartite-start} of \Cref{alg:LUT}. 

\begin{theorem}[Bian and Selinger \cite{bian2022generators}]
\label{thm:evenT_presentation}
Let $\mathcal{J}_2$ denote the ancilla-free two-qubit Clifford+$T$ operator group.
Bian and Selinger give an explicit presentation
\[
\mathcal{J}_2 \cong \langle X \mid \Gamma\rangle,
\qquad
X=\{\omega,H_0,H_1,S_0,S_1,T_0,T_1,CZ\},
\]
where the defining relations $\Gamma$ are listed in Figure~1 of their paper
(Theorem~2.1 in \cite{bian2022generators}).

Moreover, every defining relator in $\Gamma$ contains an \emph{even} total number
of occurrences of the $T$-type generators $T_0$ and $T_1$.
\end{theorem}

More useful to our current purposes and following from \Cref{thm:evenT_presentation} is the following theorem which formally justifies \Cref{line:bipartite-start} of \Cref{alg:LUT}.

\begin{theorem}
\label{thm:parity_flip_no_equal}
Let $|\cdot|_T$ be the minimal number of $T$-type generators needed to represent an
element of $\mathcal{J}_2$ (Clifford generators have zero cost), and let $t\in\{T_0,T_1\}$.
Then for every $u\in\mathcal{J}_2$,
\[
\bigl||tu|_T-|u|_T\bigr|=1.
\]
The same conclusion holds for our generator set $G$ (each $g\in G$ is a Clifford
conjugate of $T_0$ or $T_1$), i.e. for all $g\in G$ and $u\in\mathcal{J}_2$,
\[
|gu|_T \neq |u|_T,
\qquad\text{and hence}\qquad
\bigl||gu|_T-|u|_T\bigr|=1.
\]
\end{theorem}

\begin{proof}
Suppose that $u = t v$ for some shortest word $v$. Then, $\abs{u}_T \leq \abs{v}_T + 1$ and $\abs{v}_T = \abs{t^{-1} u}_T \leq \abs{u}_T+1$. Hence, 
\[
    \abs{v}_T - 1 \leq \abs{u}_T \leq \abs{v}_T + 1.
\]

Therefore, we only need to demonstrate that $\abs{tu}_T \neq \abs{u}_T$. To do this, work in the presentation $\langle X\mid\Gamma\rangle$ from
Theorem~\ref{thm:evenT_presentation}. For a word $w$ over $X$, define its
parity
\[
\pi(w) := \abs{w}_T \bmod 2 \in \mathbb{Z}_2.
\]
Here $w$ ranges over words in the free monoid on $X$; we show below that $\pi$
depends only on the group element represented by $w$.

By Theorem~\ref{thm:evenT_presentation}, every defining relator has even $T$-count. Thus, for any two words $w_0,w_1 \in \langle X \rangle$ equivalent to the same $g \in \langle X \mid \Gamma \rangle$ under the application of single relator or free reduction, we have that $\abs{\abs{w_0}_T - \abs{w_1}_T}$ is even. Hence, $\pi(w)$ is
invariant under all rewrites that preserve the represented group element.
So, $\pi$ descends to a well-defined map $\pi:\mathcal{J}_2\to\mathbb{Z}_2$.

Now let $u\in\mathcal{J}_2$ and choose any word $w$ representing $u$.
Then the concatenated word $tw$ represents $tu$ and satisfies
\[
\pi(tu)=\pi(tw)=\pi(t)+\pi(w)=1+\pi(u)\pmod 2,
\]
so $tu$ and $u$ have opposite parity. In particular, they cannot have the same
$T$-length, because every representative word for a fixed element has the same
$T$-parity, and thus the minimum $T$-count has that parity as well. Therefore
$|tu|_T\neq |u|_T$.
\end{proof}

\begin{theorem}[Correctness of Algorithm~\ref{alg:LUT}]
\label{thm:LUT_correct}
Fix $r\in \SO$ and let $\textproc{LUT}$ be the output of
\(\Call{GenerateLUT}{r,k}\) in Algorithm~\ref{alg:LUT}. Then for each
$0\le i\le k$, the set $\textproc{LUT}_i$ contains exactly one representative of
each equivalence class at $T$-distance $i$ from $r$ (and no representatives of
classes at any other $T$-distance).
\end{theorem}

\begin{proof}
We argue by loop invariant on the number of calls to
\textproc{ExtendOneStep}.

\begin{loopinvariantproof}

\Invariant
After $\ell$ iterations, the table $\textproc{LUT}$ contains layers
$\textproc{LUT}_0,\dots,\textproc{LUT}_\ell$ such that:
\begin{enumerate}[label=(\roman*)]
\item $\textproc{LUT}_i$ contains exactly one representative of each equivalence
class at $T$-distance $i$ from $r$;
\item every element of $\textproc{LUT}_i$ has $T$-distance exactly $i$.
\end{enumerate}

\Initialization
Initially, $\textproc{LUT}_0=\{r\}$ and no other layers exist.
Since $r$ has $T$-distance $0$ from itself and represents its own equivalence
class, the invariant holds for $\ell=0$.

\Maintenance
Assume the invariant holds after $\ell\ge 0$ iterations.
During the next call to \textproc{ExtendOneStep}, candidates of the form
$n=gU$ are generated for $U\in\textproc{LUT}_\ell$ and $g\in G$.
By Theorem~\ref{thm:parity_flip_no_equal}, each such candidate satisfies
$|n|_T\in\{\ell-1,\ell+1\}$.
The filter against $\textproc{LUT}_{\ell-1}$ removes precisely those candidates
at distance $\ell-1$, so all remaining elements inserted into
$\textproc{next}$ lie at distance $\ell+1$.

Every equivalence class at distance $\ell+1$ admits a representative of the form
$gU$ with $U$ at distance $\ell$.
By the inductive hypothesis, such a $U$ appears in $\textproc{LUT}_\ell$, so
\Cref{thm:equiv} guarantees that at least one representative of each class at distance $\ell+1$ is considered.
The operation \textproc{next.insert} is therefore called on at least one representative of each equivalence class and \textproc{insert} retains at most one representative per equivalence. Hence, exactly one representative is inserted into \textproc{next} per equivalence class. After setting $\textproc{LUT}_{\ell+1} \gets \textproc{next}$, the invariant
holds for $\ell+1$.

\Termination
After $k$ iterations, the invariant implies that for each
$0\le i\le k$, $\textproc{LUT}_i$ contains exactly one representative of each
equivalence class at $T$-distance $i$ from $r$.
This is precisely the claimed output of Algorithm~\ref{alg:LUT}.

\end{loopinvariantproof}
\end{proof}

\begin{theorem}[Correctness of Algorithm~\ref{alg:mitm}]
\label{thm:alg2_correct}
Algorithm~\ref{alg:mitm} terminates and returns a word of minimal $T$-length
representing $r_R c r_L^{-1}$ for some $c\in\cliff$.
\end{theorem}

\begin{proof}
We argue by loop invariant for the \textbf{while} loop in
Algorithm~\ref{alg:mitm}.

\begin{loopinvariantproof}

\Invariant
At the start of each iteration of the while loop, 
\begin{enumerate}[label=(\roman*)]
\item $\textproc{left}$ and $\textproc{right}$ are lookup tables produced by Algorithm~\ref{alg:LUT},
so that $\textproc{left}_i$ (resp.\ $\textproc{right}_j$) contains exactly one representative of each
equivalence class at $T$-distance $i$ from $r_L$ (resp.\ $j$ from $r_R$).
\item The sets of equivalence classes represented by
\[
\bigcup_{i\le d_L} \textproc{left}_i
\quad\text{and}\quad
\bigcup_{j\le d_R} \textproc{right}_j
\]
are disjoint, where $d_L$ and $d_R$ are the current $T$-counts of $L$ and $R$.
\item If an equivalence class $x$ satisfies
\[
d_T(r_L,x) + d_T(r_R,x) \le d_L + d_R,
\]
then $x$ has already been detected and the algorithm would have terminated.
\end{enumerate}

\Initialization
Initially, $L_0=\{r_L\}$ and $R_0=\{r_R\}$.
If $r_L=r_R$, the algorithm terminates immediately and returns the empty word.
Otherwise, the represented equivalence classes are disjoint and the invariant
holds with $d_L=d_R=0$.

\Maintenance
Assume the invariant holds at the start of an iteration.
Without loss of generality, suppose the algorithm extends \textproc{left} by one layer;
the case where \textproc{right} is extended is symmetric.

By Theorem~\ref{thm:LUT_correct}, the call to
\textproc{ExtendOneStep\_MITM} extends $L$ by adding exactly the equivalence
classes at $T$-distance $d_L+1$ from $r_L$, and no others.
By Theorem~\ref{thm:parity_flip_no_equal}, every such new class differs in
$T$-distance by exactly one from its parent, so no class at distance $d_L-1$
is reintroduced.

During this extension, each newly discovered equivalence class $x$ is tested
for membership in the currently deepest layer of $R$.
If $x$ is found in $R_{d_R}$, then the algorithm sets $\textproc{meet}=x$ and
terminates.
In that case, we have
\[
d_T(r_L,x)=d_L+1, \qquad d_T(r_R,x)=d_R,
\]
so $d_T(r_L,x)+d_T(r_R,x)=d_L+d_R+1$.

If no such intersection is found, then after the extension the represented
equivalence classes of $L$ up to count $d_L+1$ remain disjoint from those of $R$
up to count $d_R$, and the invariant continues to hold with $d_L$ increased by
one.

\Termination
The loop terminates at the first iteration where an equivalence class $x$ is
found in both tables.
Let $a=d_T(r_L,x)$ and $b=d_T(r_R,x)$ at termination.
Then $x$ admits representatives
\[
x = m_L r_L = m_R r_R
\]
(up to Clifford factors), where $m_L$ and $m_R$ are words of $T$-length $a$ and
$b$ respectively.

The algorithm then computes Clifford elements $c_0,c_1\in\cliff$ such that
\[
c_1 m_L r_L c_0 = m_R r_R,
\]
and returns the word $m_R^{-1} c_1 m_L r_L$, which has $T$-length $a+b$.

To see minimality, suppose there were a word representing
$r_R c r_L^{-1}$ of $T$-length strictly less than $a+b$.
Then there would exist an equivalence class $y$ with
\[
d_T(r_L,y)+d_T(r_R,y) < a+b,
\]
which would contradict part (iii) of the invariant: such a class would have been
detected earlier, forcing earlier termination.
Therefore no shorter word exists.

\end{loopinvariantproof}
\end{proof}

\section{Implementation and Optimization}\label{sec:implementation}

The theory and algorithms in \Cref{sec:theory,sec:algorithms} 
reduce exact enumeration to an extremely simple inner loop:
given a current representative $U$ and a generator $g\in G$, compute $gU$, reduce
to a canonical representative, and attempt to insert it into a set keyed by
equivalence class. Everything else is book-keeping.

That simplicity is deceptive. In practice, the runtime is dominated by the cost
of that inner loop, repeated at enormous scale. A naive implementation that
treats $U$ and $g$ as dense matrices and performs a generic $6\times 6$ matrix
multiply for each neighbor will drown in constant factors long before it hits the
interesting $T$-counts. The main engineering goal of our implementation is
therefore to \emph{redefine the primitive operations} used in
Algorithms~\ref{alg:LUT}--\ref{alg:mitm} so that each neighbor generation is
closer to ``a few cache-resident integer operations'' than to ``general linear
algebra.''

This section describes the main non-obvious implementation choices that make the
enumeration feasible. With very few exceptions, these optimizations do not
change the mathematics: they preserve exactness, the equivalence relation, and
the parity/bipartite logic. What they change is the cost of the group action
kernel, the cost of canonicalization/hashing, and the overhead of parallel
frontier management.

\subsection{Design constraints and the inner loop}
The enumeration and meet-in-the-middle routines repeatedly execute the kernel
\[
U \mapsto \mathrm{Can}(gU),
\qquad g\in G,
\]
where $\mathrm{Can}(\cdot)$ is the canonicalization procedure associated with our
equivalence relation.
This kernel sits in three performance regimes:
\begin{enumerate}
\item \textbf{Neighbor generation:} apply $g$ to $U$ quickly.
\item \textbf{Identity / equivalence filtering:} decide whether the result is
new (first by hash, then by canonical form).
\item \textbf{Data-structure overhead:} insert into a frontier set (LUT) and/or stop
early (search/MITM).
\end{enumerate}
We attacked each of these regimes separately. The main theme is specialization:
we avoid generic matrix operations, generic hash maps, and generic synchronization
in the hot path.

\subsection{Fast group action: T-steps as linear maps}
The most important optimization is the replacement of generic matrix
multiplication with specialized \emph{linear-map composition}.

\paragraph{Why generic multiplication is the wrong abstraction.}
In the $\SO$ representation, $G$ is not an arbitrary set of dense matrices.
Each generator is a Clifford conjugate of a fixed embedded $2\times 2$ rotation
(the image of $T\otimes I$ or $I\otimes T$). Recall that by \Cref{thm:involution},
\[
\overline{T\otimes I} \sim \begin{bmatrix}
    \frac{1}{\sqrt{2}} & \frac{1}{\sqrt{2}} & 0 & 0 & 0 & 0\\
    \frac{1}{\sqrt{2}} & -\frac{1}{\sqrt{2}} & 0 & 0 & 0 & 0\\
    0 & 0 & 1 & 0 & 0 & 0\\
    0 & 0 & 0 & 1 & 0 & 0\\
    0 & 0 & 0 & 0 & 1 & 0\\
    0 & 0 & 0 & 0 & 0 & 1\\
\end{bmatrix} \sim \begin{bmatrix}
    H & 0 \\
    0 & I 
\end{bmatrix}.
\]
Of course, this equivalence steps outside of $\SO$ and into $O(6)$. Nonetheless, this does not impact algorithmic performance or any claims of 
which, when we perform $\overline{T\otimes I} \;U$ for any $U \in \SO$, acts non-trivially on only the first two rows. Moreover, we know with certainty the particular integer operations that follow. Conjugation by Cliffords only changes which rows are acted on non-trivially. Hence, we implement $G_{i,j} U$ as the following linear map, exploiting \Cref{thm:involution}:
\[
    \left[C_{i,j} G_{i,j} U \right]_{k \ell} = \begin{cases}
        \frac{1}{\sqrt{2}} \left(U_{i\ell}+U_{j\ell} \right) & \text{$k=i$}\\
        \frac{1}{\sqrt{2}} \left(U_{i\ell}-U_{j\ell} \right) & \text{$k=j$}\\
        U_{k \ell} & \text{otherwise.}
    \end{cases}
\]
Hence, our implementation directly implements these as row sums/differences followed by an increment of the denominator exponent, which can be done in far fewer basic operations than naive matrix multiplication. This is exactly the computation performed by the corresponding embedded $2\times
2$ block after conjugation by a signed permutation.

\paragraph{Compile-time specialization and constant-time dispatch.}
A practical complication is that $G$ contains many conjugates, but they fall into
a small family parameterized by a pair of coordinates ($\g[i,j]$ in \Cref{sec:theory}). Instead of interpreting these indices at runtime, we generate
specialized operator implementations for each admissible coordinate pair at
compile time. At runtime, applying a generator becomes a two step process
\begin{enumerate}
\item map the generator id to a function pointer (or an index into a function
table) and
\item invoke the specialized operator.
\end{enumerate}
This keeps the hot path branch-free and lets the compiler inline the actual math.

\paragraph{Metadata-free and canonicalization-free variants.}
The enumeration uses several slightly different flavors of ``apply $g$.'' We achieve them all in one version by implementing lazy evaluations and memoization to handle canonicalization/hashing and other metadata. This is a classic trick: since these functions often go unevaluated, we only compute them once and we can direct the compiler to assume the correct branch. This keeps the instruction
stream predictable and avoids wasted work.

\subsection{Backtracking suppression and reconstruction}
Even with correct bipartite parity, naive neighbor generation still generates a
lot of redundant candidates. The simplest example is immediate backtracking:
apply a generator and then apply its inverse (or equivalently, undo the last
step). \Cref{thm:parity_flip_no_equal} already implies that a single $T$-step flips parity,
so an immediate undo is the fastest way to revisit the previous layer.

We suppress this mechanically by storing a small amount of state with each
representative: the id of the last $T$-operator used to generate it. This also provides the mechanism by which we can reconstruct a full \SO from a single lookup into our LUT generated in \Cref{alg:LUT}. 

To add efficiency, when
expanding neighbors, we skip the operator that would undo the last step. This
reduces the branching factor in the hottest loop without changing correctness:
it only removes paths that would immediately return to the layer we just came
from.

This trick is small on paper but not insignificant in practice. Because we run up against machine limitations, skipping this call can alleviate up to $\scriptstyle\frac{1}{15}$\textsuperscript{th} of our runtime. 

Because (in principle) \Cref{alg:LUT} discovers exploitable relations that should not be re-implemented, it seems plausible that a recursive data-structure might afford substantially increased efficiency. This remains a non-obvious task and for a discussion of this, see \Cref{sec:Discussion}.

\subsection{Canonicalization and hashing: cheap signatures first}
Canonicalization is conceptually simple but can be expensive if computed
from scratch, especially when unnecessary. In our implementation, canonicalization is organized as
a pipeline. First, compute a compact signature (cheap invariants). Second, use the signature to drive ordering/comparison. Finally, and only when necessary, compute more expensive canonical form data. We achieve this through a series of lazy evaluations with memoization.

\paragraph{Compact signatures.}
We store a small hash-like summary of a representative that is cheap to update
under generator application. The goal is not cryptographic hashing; it is to
avoid expensive comparisons and to keep equality checks fast in the common case. That is, we want a hash that optimizes for speed-of-computation, even at the cost of distribution quality. We do not worry as much about avoiding collisions because final comparisons always require canonical forms. This occurs because, for a hash to truly be utilized as a provable distinguisher, it would need to serve as a group homomorphism. In this case, creating a good hash seems to be as difficult as compiling itself. Our signatures/hashes are largely experimentally chosen.

\paragraph{Permutation encoding.}
Signed permutations occur constantly: they represent Cliffords in $\SO(6)$, and
they also appear as metadata inside canonicalization. We store permutations in a
compact ranked representation (e.g., Lehmer code for $S_6$) so that (1) permutation equality is a short bitstring comparison, (2) composing or inverting permutations is cheap/irrelevant, and (3) comparing canonical candidates can short-circuit early. This matters because canonicalization repeatedly compares candidates that differ only by a small permutation/sign action.

\paragraph{Allocation avoidance.}
Canonicalization often tempts one into dynamic containers (maps, vectors) for
frequency counts or support sets. However, in the inner loop, dynamic allocation is
poison. We therefore use fixed-capacity, stack-friendly structures for the
tiny maps that appear in signature computation. Our goal is to keep canonical form
and hash recomputation allocation-free in the hot path. Many of these are custom written classes that utilize lookup tables and/or optimize bitwise operations over small bitstrings.

\subsection{Frontier data structures and deduplication strategy}
The BFS layer expansion uses two kinds of deduplication: (1) \textbf{local deduplication} within the next frontier, which keeps one
representative per equivalence class; and (2) \textbf{global deduplication} against previously finalized layers, which avoids reintroducing equivalence classes that have already been explored.

\paragraph{Why global deduplication matters.}
Algorithmically, the parity theorem tells us that a one-step expansion from layer
$i$ lands in layers $i\pm 1$, so checking $\textproc{LUT.previous}$ prevents
immediate regression. However, in practice, distinct paths can still collide in the
same layer, and collisions across non-adjacent layers can occur once you quotient
by equivalence. If you only deduplicate locally, the frontier grows far faster
than the set of genuinely new classes.

Our implementation maintains a fast membership structure over the union
of finalized layers (e.g., a hash set keyed by canonical signatures). Candidate
neighbors are rejected if they appear anywhere in the explored set, not only in
the immediate predecessor layer. This is a performance optimization, not a
correctness requirement, but it drastically reduces redundant work.

\paragraph{Concurrent sets.}
The frontier expansion is embarrassingly parallel, but only if the data
structures cooperate. We represent the current layer as an iterable container
and the next layer as a concurrent hash set that supports safe parallel inserts
from many threads. The key properties we need are only that insertion (1) adds only new elements (new being determined by canonical form) and (2) returns a flag when a new element is inserted. This lets us retain only one representative per $T$-count, as well as call stop predicates only on newly inserted elements (MITM).

\subsection{Parallelization model}
The high-level parallelization strategy is simple: parallelize the expansion of
the current frontier. The details matter.

\paragraph{Work distribution.}
We use a parallel-for-each style traversal over the elements of the current
layer. Each worker:
\begin{enumerate}
\item iterates over a subset of current representatives $U$,
\item generates neighbors $gU$ for all $g\in G$ (minus backtracking skips),
\item canonicalizes and attempts insertion into the concurrent next set.
\end{enumerate}
This avoids centralized queues and makes the work naturally balanced as long as
the current layer is large.

\paragraph{Low-overhead progress tracking.}
For long runs, we track progress via per-thread counters combined at the end of
the iteration. This avoids global atomics in the hot path and gives accurate
statistics for benchmarking (insert attempts, successful inserts, stop hits).

\subsection{MITM as a stopping predicate, not a special case}
Meet-in-the-middle search uses the same underlying layer expansion as LUT
generation, but with an additional termination condition:
stop as soon as the newly discovered class intersects the opposite search.

\paragraph{Stopping predicate interface.}
Rather than embedding the MITM logic directly into the enumeration kernel, we treat
MITM as a BFS layer expansion with a user-supplied stopping predicate. The expansion routine accepts an optional predicate $\mathrm{StopPred}(x)$ that
is evaluated only when a candidate $x$ is successfully inserted into the next
frontier (i.e., it is genuinely new at that count). If the predicate returns
\textbf{\textproc{true}}, the expansion records the corresponding witness and terminates layer
generation in a best-effort way.

This is cleaner than ``break early'' in both code and exposition: (1) correctness is proved as ``the first time the predicate fires corresponds
to a minimal-count intersection,'' and (2) the implementation is just an additional boolean check in the hot path.

\paragraph{Best-effort termination and determinism.}
In parallel execution, ``stop immediately'' is not realistic without expensive
synchronization. We instead implement a standard best-effort early-stop:
once any worker finds a hit, it sets a shared stop flag and records a single
winner value. Other workers check the flag at loop boundaries and return early.
This preserves correctness (we only require that a hit implies termination, not
that termination occurs at a specific instruction) and substantially improves performance
 when a target is found early.

\paragraph{Frontier choice policy.}
In MITM, we repeatedly decide whether to expand the left or right table. A simple
and effective policy is to expand the side with the smaller current frontier. This retains balance and ensures the smallest number of $T$-multiplications in each frontier expansion. This
keeps the total work minimal for MITM. It also interacts well with the stop predicate, since a smaller frontier means fewer predicate checks.

\subsection{Memory layout, locality, and packing}
\label{sec:impl-memory-packing}
At the scale we target, memory traffic matters as much as arithmetic.
Every layer expansion touches millions of representatives and spends most of its time
(1) applying a small collection of $T$-steps, and then (2) canonicalizing and inserting results into hash tables.
We therefore design the state so that it is exact, compact, and cheap to compare.

A representative \SO consists of a $6\times 6$ matrix over $\mathbb{Z}[1/\sqrt2]$ plus a small amount of metadata used by
canonicalization and search (e.g., a cached signature). We store the 36 entries in a flat array (column-major) so the inner loop can scan columns linearly and keep cache misses predictable.

A naive approach would store each matrix entry as a floating point number or as a heap-allocated big-integer object.
Instead we represent the $x \in \mathbb{Z}[1/\sqrt{2}]$ as 
\[
x \;=\; \frac{a + b\sqrt2}{\sqrt2^{\,c}}
\qquad (a,b\in\mathbb{Z},\ c\in\mathbb{Z}_{\ge 0})
\]
in a fixed-width packed format. We store the numerator coefficients $(a,b)$ in two signed fields (two's-complement) and store the exponent $c$ in a small unsigned field. This acts like a tiny custom floating-point number, except the ``mantissa'' lives in the exact ring $\mathbb{Z}[\sqrt2]$.

Addition needs exponent alignment and we require common denominators. We split this scaling into an even part and an odd part:
\[
\sqrt2^{2m} = 2^m,
\qquad
\sqrt2^{2m+1} = \cdot 2^m\sqrt2.
\]
The $2^m$ factor is just a bit shift on the packed numerator fields.
For the remaining $\sqrt2$ factor we use the identity
\[
(a+b\sqrt2)\sqrt2 = 2b + a\sqrt2,
\]
which acts on coefficients as $(a,b)\mapsto(2b,a)$.
In packed form this becomes a constant-time ``swap-and-double'' update (swap the coefficient fields and double one of them),
so we can align exponents without unpacking or using multiword arithmetic.
This is the same coefficient-transform viewpoint as Knuth's imaginary-base arithmetic, where multiplication by $i$
acts as a swap-and-sign update $(a+bi)i=-b+ai$ \cite{Knuth1960Imaginary}.
(We do not use a positional base-$\sqrt2$ numeral system; we only borrow the constant-time coefficient transform.)

We ultimately normalize each nonzero dyadic into a \emph{reduced} encoding.
Concretely, we call $x=(a+b\sqrt2)/\sqrt2^{\,c}$ \emph{reduced} if $a$ is odd or 0 (and we encode $0$ uniquely as $a=b=c=0$). A short reduction procedure divides out the largest power of $\sqrt2$ common to the numerator coefficients and decrements $c$ accordingly. This invariant matters, because keeping memory overhead low and comparisons cheap requires a consistent single, reduced form.

The next Theorem is the reason reduction almost never runs in the hot path.
\begin{theorem}
\label{thm:dyadic-reduce-only-at-match}
Let $x=(a+b\sqrt2)/\sqrt2^{\,c}$ and $y=(a'+b'\sqrt2)/\sqrt2^{\,c'}$ be reduced nonzero dyadics, so $a$ and $a'$ are odd.
Form $x\pm y$ by using the common denominator $\sqrt2^{\max(c,c')}$. If $c\neq c'$, then the resulting encoding already has odd integer coefficient and is reduced. If $c=c'$, then the integer coefficient becomes even (odd $\pm$ odd), so at least one reduction step may be required.
\end{theorem}

\begin{proof}
Assume $c\ge c'$ and set $N=c-c'$.
If $N=0$, then $x\pm y=((a\pm a')+(b\pm b')\sqrt2)/\sqrt2^{\,c}$ and the integer coefficient $a\pm a'$ is even since $a,a'$ are odd.

Now, without loss of generality, assume $N>0$. We scale $y$ by $\sqrt2^{N}$. If $N=2m$ is even, then $\sqrt2^{N}=2^m$, so the scaled integer coefficient is $2^m a'$, which is even. If $N=2m+1$ is odd, then
\[
(a'+b'\sqrt2)\sqrt2^{N} = (2b'+a'\sqrt2)\,2^m,
\]
so the scaled integer coefficient is $2^{m+1}b'$ and is, again, even. In either case, the summand coming from $y$ has an even integer coefficient, while $x$ has an odd integer coefficient. Therefore, the integer coefficient of the sum/difference remains odd, so the result is already reduced.
\end{proof}

Two's-complement fields make sign operations cheap. Unary negation multiplies the entry by $-1$ by applying a single modular complement to the packed numerator word (with a special case so that zero has a unique encoding),
so it negates both $a$ and $b$ at once.
The nontrivial Galois automorphism of $\mathbb{Q}(\sqrt2)$,
\[
(a+b\sqrt2)^\bullet := a-b\sqrt2,
\]
just flips the sign of the $b$ field.
Exact Clifford+$T$ synthesis papers working in cyclotomic/dyadic rings often denote this map by a ``bullet'' and refer to it as $\sqrt2$-conjugation or a \emph{twist}
\cite{giles2013exact,kliuchnikov2013synthesis,selinger2012efficient}. We have done this both through bit-twiddling and direct negation on $b$, though after many code optimizations, bit-twiddling proved less useful.

Packing makes equality and hashing cheap. We can compare entries by comparing their packed words, and we can hash them by hashing those same bits. At the matrix level we add an extra short-circuit layer: we maintain small cached signatures (hash-like summaries) that update cheaply under $T$-steps. Most candidate insertions fail quickly on these signatures without touching all 36 entries.

Canonicalization and hashing still dominate runtime, so we structure them to be kind to caches.
Whenever possible we update signatures and small frequency summaries incrementally during generator application, instead of recomputing them from scratch.
When we do recompute hashes, we scan the matrix in the stored order and avoid data-dependent indirections. The guiding rule is simple: touch as few entries as possible, and when you must touch them, touch them sequentially.

At the scale we target, memory traffic is as important as arithmetic. We have reduced our memory footprint through a number of measures.

\paragraph{Locality-aware canonicalization.}
Whenever possible, we update signatures and metadata lazily during
generator application. When recomputation is necessary, we structure it to read
data sequentially and to keep working sets small. Our biggest win is simply to avoid  touching all 36 entries unless truly necessary.

\subsection{Summary: why these optimizations matter}
The paper’s algorithms are simple, but they only become useful when the constant
factors are small. The optimizations above reduce those constant factors in three
ways:
\begin{enumerate}
\item They replace generic matrix multiplication with specialized linear-map
kernels for $T$-steps.
\item They reduce the number of candidate neighbors that survive long enough to
access expensive data structures.
\item They make parallel exploration inexpensive by using concurrent frontiers, atomic
best-effort stopping, and allocation-free hot paths.
\end{enumerate}
The empirical effect of these choices appears directly in the scaling behavior
reported in \Cref{sec:Results}: once the inner loop is cheap enough,
the enumeration becomes limited by the intrinsic growth of the search space
rather than by avoidable overhead.

\section{Benchmarking and Results}\label{sec:Results}
\subsection{Experimental setup}
\label{sec:experimental-setup}
We ran all experiments on a single-socket workstation running Ubuntu~25.04 (Plucky Puffin) with Linux kernel \texttt{6.14.0-37-generic}. 
The machine uses an AMD Ryzen 9~7900X CPU (12 physical cores / 24 hardware threads) with a single NUMA node. The system reports a nominal frequency of 4.70\,GHz with boost support enabled.
Cache sizes are 32\,KiB L1d and 32\,KiB L1i per core, 1\,MiB L2 per core, and 64\,MiB shared L3.
The machine has 192\,GB of RAM with 8\,GiB swap configured.
All runs used local NVMe storage; the working filesystem is ext4 on a 2\,TB Samsung 9100~PRO SSD.

\subsection{Benchmarks}
In this section, we present our various benchmarks. \Cref{tab:lut_mem} includes benchmarking results for lookup table generation for both the standard (identity) root and a randomly chosen root. We plot these results in \Cref{tab:lut_mem}. These benchmarks were derived from the codebase found in \cite{Jarret_Exact_Synthesis}.

{
\begin{table}[t]
\centering
\scriptsize
\setlength{\tabcolsep}{2.5pt}
\caption{LUT build scaling (Google Benchmark output). Columns are maximum depth $k$. \texttt{reps} is the total number of stored canonical representatives across all layers $0..k$ (cumulative). Times are per-iteration CPU time and RSS is peak representative set size.}
\label{tab:lut_mem}
\begin{tabular}{lrrrrrrrrrrrr}
\toprule
 & 1 & 2 & 3 & 4 & 5 & 6 & 7 & 8 & 9 & 10 & 11 & 12 \\
\midrule
\multicolumn{13}{l}{\textbf{identity}} \\
representatives & 2 & 4 & 10 & 29 & 106 & 477 & 2,637 & 17,092 & 123,925 & 970,266 & 8.0e6 & 6.7e7 \\
time (s) & 0.002 & 0.003 & 0.004 & 0.007 & 0.01 & 0.012 & 0.018 & 0.036 & 0.135 & 0.964 & 11.1 & 414 \\
RSS (GB) & 0.005 & 0.005 & 0.007 & 0.008 & 0.009 & 0.01 & 0.012 & 0.022 & 0.07 & 0.362 & 2.65 & 22.33 \\
\midrule
\multicolumn{13}{l}{\textbf{random}} \\
representatives & 16 & 181 & 1,919 & 20,073 & 208,142 & $2.1e6$ & $2.2e7$ & $2.3e8$ & -- & -- & -- & -- \\
time (s) & 6.7 e{-6} & 1e-4 & 0.000731 & 0.007 & 0.094 & 1.5 & 48.8 & 4.74e3 & -- & -- & -- & -- \\
RSS (GB) & .005 & .007 & .010 & .023 & .109 & .741 & 7.15 & 65.7 & -- & -- & -- & -- \\
\bottomrule
\end{tabular}
\end{table}
}

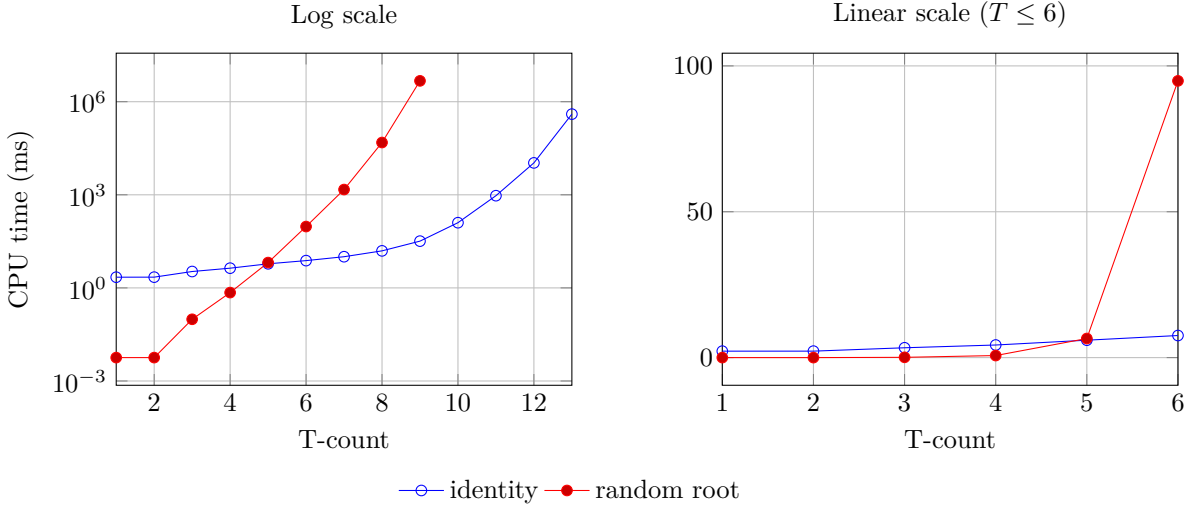
\begin{figure}[H]
\centering
\begin{tikzpicture}
\begin{groupplot}[
  group style={group size=2 by 1, horizontal sep=2cm},
  width=0.49\linewidth,
  height=6cm,
  grid=both,
  xlabel={T-count},
  ylabel={CPU time (ms)}
]

\nextgroupplot[
  title={Log scale},
  xmin=1, xmax=13,
  ymode=log,
  clip=false,
  legend style={
    at={(1,-0.25)}, anchor=north,
    draw=none,
    legend columns=3,
  },
]
\addplot+[mark=o] coordinates {
  (1,2.212093716) (2,2.22038592) (3,3.373862858) (4,4.320413855)
  (5,5.970639387) (6,7.57404067) (7,10.11660853) (8,15.62299899)
  (9,32.13974656) (10,126.2976408) (11,936.441558) (12,10683.06666)
  (13,398324.9451)
};
\addlegendentry{identity}

\addplot+[mark=*] coordinates {
  (1,5.669405868e-3) (2,5.672614085e-3) (3,9.739437177e-2)
  (4,7.060268914e-1) (5,6.54640413) (6,94.87038023)
  (7,1476.831323) (8,48408.01546) (9,4700958.091)
};
\addlegendentry{random root}

\nextgroupplot[
  title={Linear scale ($T\le6$)},
  xmin=1, xmax=6,
  ylabel={}, %
]
\addplot+[mark=o] coordinates {
  (1,2.212093716) (2,2.22038592) (3,3.373862858)
  (4,4.320413855) (5,5.970639387) (6,7.57404067)
};
\addplot+[mark=*] coordinates {
  (1,5.669405868e-3) (2,5.672614085e-3) (3,9.739437177e-2)
  (4,7.060268914e-1) (5,6.54640413) (6,94.87038023)
};

\end{groupplot}
\end{tikzpicture}

\caption{\label{fig:lut_scaling} Time to build a complete lookup table of all unique representatives (up to equivalence) for a particular T-count.
\label{fig:lutbuild}}
\end{figure}
 
As one can see, the random root rapidly approaches its asymptotic scaling, which should never grow substantially faster than powers of $14$. At large enough $T$-counts, the identity-rooted table also begins to scale as the randomly-rooted table. Nonetheless, at the outset random roots greatly outperform the identity.

This behavior is not difficult to understand. With a random root, each $T$ operator applied to the root is very likely to create an easily distinguishable representative. Thus, we rarely decide equivalence and can even rapidly reject equivalences allowing immediate inserts. At around $T=5$, we see that the absolute size of the relative database (and hence the sheer number of comparisons) begins to compete with the speed with which each comparison can be performed. 

In contrast, with the identity as a root, the first few $T$-counts have a high degree of equivalences. (All single $T$ \SO matrices are equivalent, for instance.) Thus, the growth rate of the table is slower, as evidenced by the longer lead time before exponential explosion. At higher $T$ counts, we eventually end up in a similar situation as the random root; easily distinguishable representatives, but large databases. Of course, the database size of the identity-rooted table does ultimately catch up to that of the randomly-rooted table.

Next, we consider the performance of our meet-in-the-middle routine. We plot our performance against \cite{Amy2013}, which we compiled and ran locally for reference. We consider two versions of the meet-in-the-middle procedure. In the standard version, we run \Cref{alg:mitm} directly as stated. In the brute force amended version, after each new layer of the LUT is created, we do a brute force search for intersections, where we store nothing in memory. The time to perform brute force search is left as a hardcoded parameter, but a natural choice would be a length of time proportionate to the length of time one anticipates it will take to create the next layer. This could, of course, be calculated, but we do not do that here. 

\begin{figure}[H]
\centering
\begin{tikzpicture}
\begin{groupplot}[
  group style={group size=2 by 1, horizontal sep=2cm},
  width=0.49\linewidth,
  height=6cm,
  grid=both,
  xlabel={T-count},
  ylabel={CPU time (ms)}
]

\nextgroupplot[
  title={Log scale},
  ymode=log,
  clip=false,
  legend style={
    at={(1,-0.25)}, anchor=north,
    draw=none,
    legend columns=2,
  },
]

\addplot+[mark=o] coordinates {
  (1,6.38136)
  (2,293.365)
  (3,379.378)
  (4,470.216)
  (5,276.206)
  (6,697.955)
  (7,393.638)
  (8,437.673)
  (9,517.395)
  (10,587.806)
  (11,532.576)
  (12,751.049)
  (13,1048.02)
  (14,2027.67)
  (15,4244.28)
  (16,15181)
  (17,34703.8)
  (18,165059)
};
\addlegendentry{standard}

\addplot+[mark=*] coordinates {
  (1,6.36778)
  (2,24.5729)
  (3,24.2383)
  (4,25.5553)
  (5,36.395)
  (6,236.797)
  (7,437.79)
  (8,463.228)
  (9,578.31)
  (10,820.077)
  (11,953.365)
  (12,1358.74)
  (13,1833.99)
  (14,2521.03)
  (15,4923.19)
  (16,7938.42)
  (17,33719.8)
  (18,96831.9)
  (20,3.13455e6)
};
\addlegendentry{standard + brute force}

\addplot+[mark=x] coordinates {
  (1,30.2)
  (2,335)
  (3,339)
  (4,2921)
  (5,2947)
  (6,24171)
  (7,175864)
};
\addlegendentry{Amy et al. \cite{mitms2026} (\textit{not} $T$-) depth-optimal}

\nextgroupplot[
  title={Linear scale ($T \le 18$)},
  ymode=linear,
  ylabel={}, 
  xmin=1,
  xmax=18,
  scaled y ticks=false,
  clip=false
  ]

\addplot+[mark=o] coordinates {
  (1,6.38136)
  (2,293.365)
  (3,379.378)
  (4,470.216)
  (5,276.206)
  (6,697.955)
  (7,393.638)
  (8,437.673)
  (9,517.395)
  (10,587.806)
  (11,532.576)
  (12,751.049)
  (13,1048.02)
  (14,2027.67)
  (15,4244.28)
  (16,15181)
  (17,34703.8)
  (18,165059)
};

\addplot+[mark=*] coordinates {
  (1,6.36778)
  (2,24.5729)
  (3,24.2383)
  (4,25.5553)
  (5,36.395)
  (6,236.797)
  (7,437.79)
  (8,463.228)
  (9,578.31)
  (10,820.077)
  (11,953.365)
  (12,1358.74)
  (13,1833.99)
  (14,2521.03)
  (15,4923.19)
  (16,7938.42)
  (17,33719.8)
  (18,96831.9)
};

\addplot+[mark=x] coordinates {
  (1,30.2)
  (2,335)
  (3,339)
  (4,2921)
  (5,2947)
  (6,24171)
  (7,175864)
};

\end{groupplot}
\end{tikzpicture}

\caption{Time to perform MITM search with no pre-built database for a representative. The horizontal axis is the size of an optimal $T$-count circuit for that representative. We plot against the depth-optimized and \textit{not} $T$-depth optimized circuits found by \cite{Amy2013}. $T$-depth optimized circuits were too slow to be plotted on the same scale.}
\end{figure}
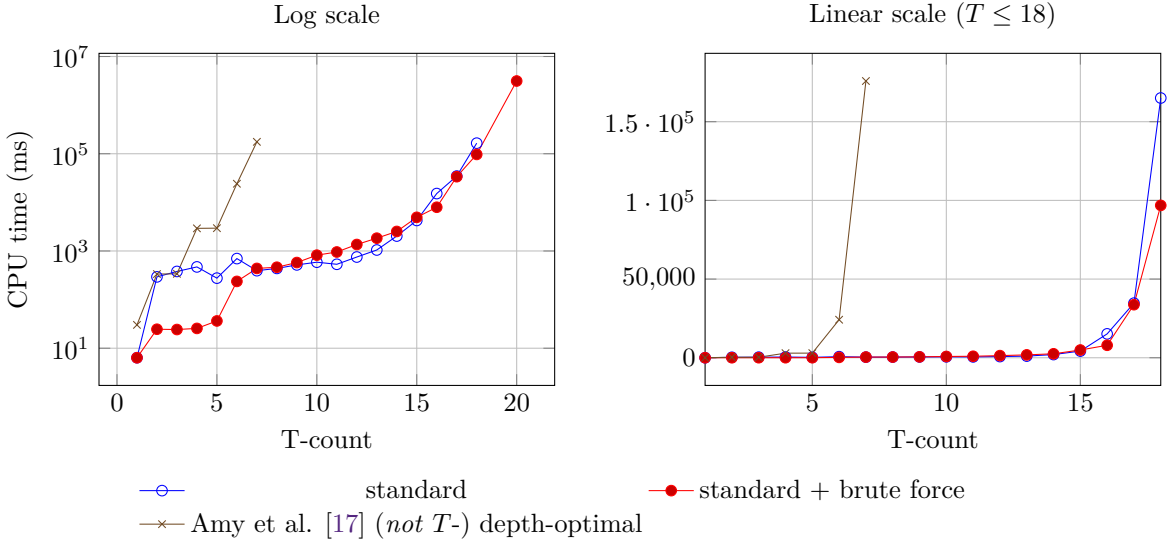
 
In this case, what we find is that the two methods scale fairly similarly initially. At higher $T$-counts, the brute force amended version begins to become highly advantageous, allowing us to reach higher $T$-counts than previously possible. We anticipate that some playing with and balancing of this parameter can allow us to perform deeper exact synthesis much faster, since parallelization becomes more useful here, especially in HPC settings.

Most importantly is the comparison of our methods to those of Amy et al. Importantly, we \textit{did not} compare to their $T$-depth-optimal procedure (which is the closest approximant of $T$-count allowed by their codebase) to ours. This procedure performed too slowly to plot against our own methods. Rather, we compared depth-optimal, end-to-end (LUT generation and MITM) compiling directly to our own performance. One can immediately see that, despite exponential scaling, we are able to compile $T$-count optimal circuits at what would have previously been intractable $T$-counts. 

\section{Discussion and Future Work}
\label{sec:Discussion}
This paper argues that exact synthesis can be practical as compiler infrastructure when (i) the target regime is structured and repeatedly encountered, and (ii) the implementation is engineered so that enumeration and equivalence reduction are limited by intrinsic search-space growth rather than avoidable constant factors. Although we instantiate the approach for two-qubit Clifford+T with T-count as the cost metric, the broader lesson is that combining strong equivalence reduction with gate-set-aware arithmetic can turn offline exact enumeration into a reusable backend for local optimization.

\subsection{Canonicalization as the primary bottleneck}

At present, canonicalization and equivalence checking are among the dominant costs in both lookup-table generation and meet-in-the-middle search. Our current strategy is correct but not yet optimal, and there are clear opportunities to reduce this overhead without changing the underlying mathematics.

First, our implementation currently explores a large signed-permutation action during canonicalization. Because the equivalence relation and the $\SO$ embedding introduce a natural distinction between determinant-preserving and determinant-flipping signed permutations, restricting to an appropriate subgroup (or exploiting a parity/involution structure induced by representative choices) should reduce canonicalization cost and improve scaling.

Second, the involutive-generator viewpoint suggests that tighter normal forms may exist, potentially allowing early short-circuiting of canonicalization in many cases. Developing such structure-aware canonical forms would not change correctness, but could reduce the amortized cost of equivalence testing and thereby increase feasible depth.

\subsection{Extending the step set and constant-time kernels}

A major practical advantage comes from implementing T-steps as specialized linear maps rather than generic matrix multiplication. Our choice of step operators is not unique. An important direction is to enlarge the step set while preserving constant-time application (for example, by incorporating short non-Clifford motifs that admit the same arithmetic shortcuts). Even constant-factor changes to branching and neighbor generation can yield substantial improvements in reachable depth.

\subsection{Large-scale lookup tables and dataset generation}

The lookup-table approach is most powerful when treated as an offline resource that can be generated once and reused across many compilation tasks. Because frontier expansion is embarrassingly parallel, cluster-scale runs are a natural next step for expanding coverage. An enlarged lookup table need not contain all shortest circuits; rather, it should provide broad coverage of equivalence classes with optimal representatives, enabling fast query-based synthesis and supplying large ground-truth datasets for training and evaluating heuristic compilers. 

\subsection{Integration into end-to-end compilation workflows}

For broad adoption, two engineering tasks matter. First, end-user toolchains will require robust conversion between standard $SU(4)/U(4)$ descriptions and the internal $\SO$ representation used here, so that the lookup table and meet-in-the-middle engine can be invoked as backend routines. Second, integrating this exact backend into heuristic pipelines (e.g., sliding-window passes) will require interfaces for matching and rewriting subcircuits and for tracking Clifford frames efficiently.

\subsection{Beyond Clifford+T and beyond SO(6)}

Although Clifford+T with T-count is a natural initial target, the separation between an outer enumeration/MITM layer and an inner arithmetic layer makes the approach portable: the outer logic remains unchanged as long as the inner layer provides exact representations and efficient equivalence testing. Finally, several of our speedups derive from packed exact arithmetic rather than from $\SO$ structure per se; it is therefore plausible that a direct $SU(4)$ implementation with comparable packing could be competitive or better in practice. Such an implementation would potentially use fewer basic operations in each generating step. Thus, it may have been a mistake to utilize the $\SO$ representation and future work may revert to $SU(4)$.
 
\section{Conclusion}

Optimal exact synthesis is valuable because it produces unconditional optima, canonical representatives, and ground-truth instances for evaluating and improving heuristic compilers. However, exact methods are typically limited to carefully scoped regimes, both by intrinsic search-space growth and by the constant factors associated with exact arithmetic and equivalence testing. In this work, we showed that a ``pay once, query forever'' model can be made practical for short-enough two-qubit Clifford+T compilation when optimizing \emph{T-count}, by exhaustively enumerating a bounded region of the operator group while quotienting by compiler-natural equivalences.

Our main contribution is a synthesis pipeline that is simultaneously \emph{exact} and \emph{high-performance}. Algorithmically, we combine breadth-first enumeration with meet-in-the-middle joining and provably correct pruning rules derived from the parity (bipartite) structure induced by T-steps. Representationally, we operate in an $SO(6)$ model in which multiplication, comparison, and Clifford actions admit efficient exact implementations. Practically, performance comes from an inner arithmetic layer that replaces general-purpose linear algebra with specialized data structures and constant-time kernels, so that the hot loop is dominated by cache-resident integer operations rather than dense matrix routines. The resulting implementation enables lookup-table generation and exact meet-in-the-middle synthesis at T-counts that are impractical for baseline exact-synthesis implementations on comparable hardware.

This work is intentionally scoped. It does not circumvent worst-case hardness results for exact non-Clifford optimization, and it is not intended to replace heuristic compilation for large-scale circuits. Instead, it provides reusable, provably optimal two-qubit building blocks and an engine for generating additional ground-truth instances, both of which are directly applicable as backends for heuristic optimizers. The principal limitation at present is that canonicalization and equivalence checking remain dominant costs at larger depths; further progress will rely on more efficient canonical forms and quotienting strategies.

Looking forward, two directions are immediate. First, improving canonicalization and scaling lookup-table generation (including distributed generation) will increase the practical reach of offline exact synthesis as a reusable resource. Second, the separation between an outer enumeration/MITM layer and an inner gate-set-aware arithmetic layer suggests a path to transferring these methods to other discrete gate sets and cost metrics whenever exact representations and efficient equivalence testing are available. More broadly, our results support a compilation model in which small, exhaustively characterized subroutines serve as reusable infrastructure inside hybrid compiler pipelines.

\section{Acknowledgements}\label{sec:Acknowledgments}
The authors thank John Weston for early code contributions and Stephen Jordan for helpful discussions. We also thank the Mason Experimental Geometry Lab (MEGL) for computational support, and the MEGL students Mark Dubynskyi and Safiuddeen Salem for their dedicated efforts in exploratory and preliminary computational work.

This material is based upon work supported by the National Science Foundation under Grant No. CCF-2038024. Any opinions, findings, and conclusions or recommendations expressed in this material are those of the author(s) and do not necessarily reflect the views of the National Science Foundation.
This project was supported by resources provided by the Office of Research Computing at George Mason University (URL: https://orc.gmu.edu) and funded in part by grants from the National Science Foundation (Award Number 2018631).

\section{Statement on Generative AI}\label{sec:AI}
In authoring this manuscript, the authors exploited the use of ChatGPT. The intellectual content was derived without AI. Generative AI was used to generate and clean up both prose and code under close supervision. Additionally, AI was used to unify notation and clarify theorems, so some artifacts of its use may persist.

\bibliographystyle{apsrev4-2}
\bibliography{bibliography}

\end{document}